\documentclass{article}%
\usepackage{amsmath}
\usepackage{amssymb}
\usepackage{amsfonts}
\usepackage{graphicx}%
\providecommand{\U}[1]{\protect\rule{.1in}{.1in}}
\newtheorem{theorem}{Theorem} [section]

\newtheorem{corollary}[theorem]{Corollary}

\newtheorem{lemma}[theorem]{Lemma}

\newtheorem{problem}[theorem]{Problem}
\newtheorem{proposition}[theorem]{Proposition}

\newenvironment{proof}[1][Proof]{\noindent\textbf{#1.} }{\ \rule{0.5em}{0.5em}}
\begin{document}

\title{On Symmetry of Independence Polynomials}
\author{Vadim E. Levit\\Ariel University Center of Samaria, Israel\\levitv@ariel.ac.il
\and Eugen Mandrescu\\Holon Institute of Technology, Israel\\eugen\_m@hit.ac.il}
\date{}
\maketitle

\begin{abstract}
An \textit{independent} set in a graph is a set of pairwise non-adjacent
vertices, and $\alpha(G)$ is the size of a maximum independent set in the
graph $G$. A matching is a set of non-incident edges, while $\mu(G)$ is the
cardinality of a maximum matching.

If $s_{k}$ is the number of independent sets of cardinality $k$ in $G$, then
\[
I(G;x)=s_{0}+s_{1}x+s_{2}x^{2}+...+s_{\alpha}x^{\alpha},\alpha=\alpha\left(
G\right)  ,
\]
is called the \textit{independence polynomial} of $G$ (Gutman and Harary
\cite{GuHa83}). If $s_{j}=s_{\alpha-j}$, $0\leq j\leq\left\lfloor
\alpha/2\right\rfloor $, then $I(G;x)$ is called \textit{symmetric} (or
\textit{palindromic}). It is known that the graph $G\circ2K_{1}$ obtained by
joining each vertex of $G$ to two new vertices, has a symmetric independence
polynomial \cite{St98}.

In this paper we show that for every graph $G$ and for each non-negative
integer $k\leq\mu\left(  G\right)  $, one can build a graph $H$, such that:
$G$ is a subgraph of $H$, $I\left(  H;x\right)  $ is symmetric, and $I\left(
G\circ2K_{1};x\right)  =\left(  1+x\right)  ^{k}\cdot I\left(  H;x\right)
$.\medskip

\medskip

\textbf{MSC Classification 2010:} 05C31; 05C69.

\end{abstract}

\section{Introduction}

Throughout this paper $G=(V,E)$ is a simple (i.e., a finite, undirected,
loopless and without multiple edges) graph with vertex set $V=V(G)$ and edge
set $E=E(G).$ If $X\subset V$, then $G[X]$ is the subgraph of $G$ spanned by
$X$. By $G-W$ we mean the subgraph $G[V-W]$, if $W\subset V(G)$. We also
denote by $G-F$ the partial subgraph of $G$ obtained by deleting the edges of
$F$, for $F\subset E(G)$, and we write shortly $G-e$, whenever $F$ $=\{e\}$.
The \textit{neighborhood} of a vertex $v\in V$ is the set $N_{G}(v)=\{w:w\in
V$ \textit{and }$vw\in E\}$, and $N_{G}[v]=N_{G}(v)\cup\{v\}$; if there is no
ambiguity on $G$, we write $N(v)$ and $N[v]$. $K_{n},P_{n},C_{n}$ denote,
respectively, the complete graph on $n\geq1$ vertices, the chordless path on
$n\geq1$ vertices, and the chordless cycle on $n\geq3$ vertices.

The \textit{disjoint union} of the graphs $G_{1},G_{2}$ is the graph
$G=G_{1}\cup G_{2}$ having as vertex set the disjoint union of $V(G_{1}%
),V(G_{2})$, and as edge set the disjoint union of $E(G_{1}),E(G_{2})$. In
particular, $nG$ denotes the disjoint union of $n>1$ copies of the graph $G$.

If $G_{1},G_{2}$ are disjoint graphs, $A_{1}\subseteq V(G_{1}),A_{2}\subseteq
V(G_{2})$, then the \textit{Zykov sum} of $G_{1},G_{2}$ with respect to
$A_{1},A_{2}$, is the graph $\left(  G_{1},A_{1}\right)  +(G_{2},A_{2})$ with
$V(G_{1})\cup V(G_{2})$ as vertex set and $E(G_{1})\cup E(G_{2})\cup
\{v_{1}v_{2}:v_{1}\in A_{1},v_{2}\in A_{2}\}$ as edge set. If $A_{1}=V(G_{1})$
and $A_{2}=V(G_{2})$, we simply write $G_{1}+G_{2}$.

The \textit{corona} of the graphs $G$ and $H$ with respect to $A\subseteq
V(G)$ is the graph $\left(  G,A\right)  \circ H$ obtained from $G$ and
$\left\vert A\right\vert $ copies of $H$, such that each vertex of $A$ is
joined to all vertices of a copy of $H$. If $A=$ $V(G)$ we use $G\circ H$
instead of $\left(  G,V(G)\right)  \circ H$ (see Figure \ref{fig123} for an
example). \begin{figure}[h]
\setlength{\unitlength}{1cm}\begin{picture}(5,1.2)\thicklines
\put(1,0){\circle*{0.29}}
\multiput(1,1)(1,0){4}{\circle*{0.29}}
\put(1,0){\line(1,1){1}}
\put(1,0){\line(0,1){1}}
\put(1,1){\line(1,0){3}}
\put(3,0.6){\makebox(0,0){$a$}}
\put(4,0.6){\makebox(0,0){$b$}}
\put(0.3,0.5){\makebox(0,0){$G$}}
\multiput(5,1)(1,0){2}{\circle*{0.29}}
\put(5,1){\line(1,0){1}}
\put(5.5,0.5){\makebox(0,0){$H$}}
\multiput(8,0)(1,0){5}{\circle*{0.29}}
\multiput(8,1)(1,0){4}{\circle*{0.29}}
\put(8,0){\line(0,1){1}}
\put(8,0){\line(1,1){1}}
\put(8,1){\line(1,0){3}}
\put(9,0){\line(1,1){1}}
\put(9,0){\line(1,0){1}}
\put(10,0){\line(0,1){1}}
\put(11,0){\line(1,0){1}}
\put(11,0){\line(0,1){1}}
\put(11,1){\line(1,-1){1}}
\put(10.25,0.7){\makebox(0,0){$a$}}
\put(11.35,1){\makebox(0,0){$b$}}
\put(7.3,0.5){\makebox(0,0){$L$}}
\end{picture}
\caption{$G,H$ and $L=\left(  G,A\right)  \circ H$, where $A=\{a,b\}$.}%
\label{fig123}%
\end{figure}
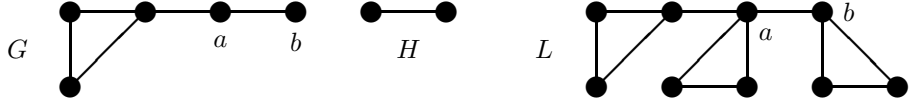

Let $G,H$ be two graphs and $C$ be a cycle on $q$ vertices of $G$. By
$(G,C)\bigtriangleup H$ we mean the graph obtained from $G$ and $q$ copies of
$H$, such that each two consecutive vertices on $C$ are joined to all vertices
of a copy of $H$ (see Figure \ref{fig333} for an example).

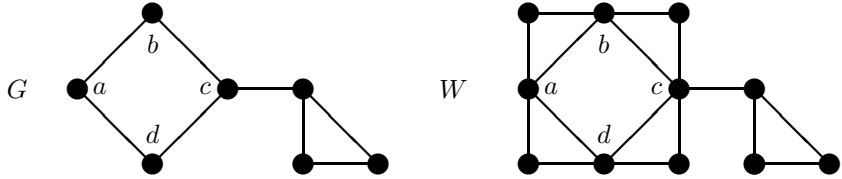
\begin{figure}[h]
\setlength{\unitlength}{1cm}\begin{picture}(5,2.1)\thicklines
\put(2,1){\circle*{0.29}}
\multiput(3,0)(0,2){2}{\circle*{0.29}}
\multiput(4,1)(1,0){2}{\circle*{0.29}}
\multiput(5,0)(1,0){2}{\circle*{0.29}}
\put(2,1){\line(1,1){1}}
\put(2,1){\line(1,-1){1}}
\put(3,0){\line(1,1){1}}
\put(3,2){\line(1,-1){1}}
\put(4,1){\line(1,0){1}}
\put(5,0){\line(1,0){1}}
\put(5,0){\line(0,1){1}}
\put(5,1){\line(1,-1){1}}
\put(2.3,1){\makebox(0,0){$a$}}
\put(3,1.6){\makebox(0,0){$b$}}
\put(3.7,1){\makebox(0,0){$c$}}
\put(3,0.4){\makebox(0,0){$d$}}
\put(1.2,1){\makebox(0,0){$G$}}
\multiput(8,0)(0,1){3}{\circle*{0.29}}
\multiput(9,0)(0,2){2}{\circle*{0.29}}
\multiput(10,0)(0,1){3}{\circle*{0.29}}
\multiput(11,0)(0,1){2}{\circle*{0.29}}
\multiput(11,0)(1,0){2}{\circle*{0.29}}
\put(8,0){\line(0,1){2}}
\put(8,0){\line(1,0){2}}
\put(8,1){\line(1,1){1}}
\put(8,1){\line(1,-1){1}}
\put(8,2){\line(1,0){2}}
\put(9,0){\line(1,1){1}}
\put(9,2){\line(1,-1){1}}
\put(10,0){\line(0,1){2}}
\put(10,1){\line(1,0){1}}
\put(11,0){\line(1,0){1}}
\put(11,0){\line(0,1){1}}
\put(11,1){\line(1,-1){1}}
\put(8.3,1){\makebox(0,0){$a$}}
\put(9,1.6){\makebox(0,0){$b$}}
\put(9.7,1){\makebox(0,0){$c$}}
\put(9,0.4){\makebox(0,0){$d$}}
\put(7,1){\makebox(0,0){$W$}}
\end{picture}
\caption{$G$ and $W=(G,C)\bigtriangleup H$, where $V(C)=\{a,b,c,d\}$ and
$H=K_{1}$.}%
\label{fig333}%
\end{figure}

An \textit{independent} (or a \textit{stable})\textit{\ }set in $G$ is a set
of pairwise non-adjacent vertices. By \textrm{Ind}$(G)$ we mean the family of
all independent sets of $G$. An independent set of maximum size will be
referred to as a \textit{maximum independent set} of $G$, and the
\textit{independence number }of $G$, denoted by $\alpha(G)$, is the
cardinality of a maximum independent set in $G$.

Let $s_{k}$ be the number of independent sets of size $k$ in a graph $G$. The
polynomial
\[
I(G;x)=s_{0}+s_{1}x+s_{2}x^{2}+...+\text{ }s_{\alpha}x^{\alpha},\quad
\alpha=\alpha\left(  G\right)  ,
\]
is called the \textit{independence polynomial} of $G$ \cite{GuHa83}, the
\textit{independent set polynomial} of $G$ \cite{HoLi94}. In \cite{FiSo90},
the \textit{dependence polynomial} $D(G;x)$ of a graph $G$ is defined as
$D(G;x)=I(\overline{G};-x)$. 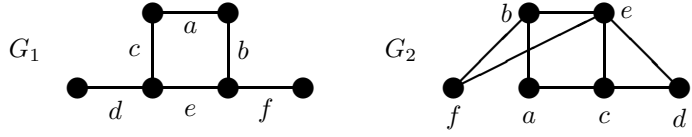
\begin{figure}[h]
\setlength{\unitlength}{1cm}\begin{picture}(5,1.8)\thicklines
\multiput(3,0.5)(1,0){4}{\circle*{0.29}}
\multiput(4,1.5)(1,0){2}{\circle*{0.29}}
\put(3,0.5){\line(1,0){3}}
\put(4,0.5){\line(0,1){1}}
\put(4,1.5){\line(1,0){1}}
\put(5,0.5){\line(0,1){1}}
\put(3.75,1){\makebox(0,0){$c$}}
\put(5.2,1){\makebox(0,0){$b$}}
\put(3.5,0.2){\makebox(0,0){$d$}}
\put(4.5,1.3){\makebox(0,0){$a$}}
\put(5.5,0.2){\makebox(0,0){$f$}}
\put(4.5,0.2){\makebox(0,0){$e$}}
\put(2.3,1){\makebox(0,0){$G_{1}$}}
\multiput(8,0.5)(1,0){4}{\circle*{0.29}}
\multiput(9,1.5)(1,0){2}{\circle*{0.29}}
\put(8,0.5){\line(1,1){1}}
\put(8,0.5){\line(2,1){2}}
\put(9,0.5){\line(1,0){2}}
\put(9,0.5){\line(0,1){1}}
\put(9,1.5){\line(1,0){1}}
\put(10,0.5){\line(0,1){1}}
\put(10,1.5){\line(1,-1){1}}
\put(8,0.1){\makebox(0,0){$f$}}
\put(8.7,1.5){\makebox(0,0){$b$}}
\put(10,0.1){\makebox(0,0){$c$}}
\put(11,0.1){\makebox(0,0){$d$}}
\put(9,0.1){\makebox(0,0){$a$}}
\put(10.3,1.5){\makebox(0,0){$e$}}
\put(7.3,1){\makebox(0,0){$G_{2}$}}
\end{picture}
\caption{$G_{2}$ is the line-graph of and $G_{1}$.}%
\label{fig1}%
\end{figure}

A matching is a set of non-incident edges of a graph $G$, while $\mu(G)$ is
the cardinality of a maximum matching. Let $m_{k}$ be the number of matchings
of size $k$ in $G$. The polynomial
\[
M(G;x)=m_{0}+m_{1}x+m_{2}x^{2}+...+\text{ }m_{\mu}x^{\mu},\quad\mu=\mu\left(
G\right)  ,
\]
is called the \textit{matching polynomial} of $G$ \cite{Farel79}.

The independence polynomial has been defined as a generalization of the
matching polynomial, because the matching polynomial of a graph $G$ and the
independence polynomial of its line graph are identical. Recall that given a
graph $G$, its \textit{line graph} $L(G)$ is the graph whose vertex set is the
edge set of $G$, and two vertices are adjacent if they share an end in $G$.
For instance, the graphs $G_{1}$ and $G_{2}$ depicted in Figure \ref{fig1}
satisfy $G_{2}=L(G_{1})$ and, hence, $I(G_{2};x)=1+6x+7x^{2}+x^{3}=M(G_{1};x)$.

In \cite{GuHa83} a\emph{\ }number of general properties of the independence
polynomial of a graph are presented. As examples, we mention that:%
\begin{align*}
I(G_{1}\cup G_{2};x)  &  =I(G_{1};x)\cdot I(G_{2};x),\\
I(G_{1}+G_{2};x)  &  =I(G_{1};x)+I(G_{2};x)-1.
\end{align*}
The following equalities are very useful in calculating of the independence
polynomial for various families of graphs.

\begin{theorem}
\label{th3}\emph{(i)} \cite{GuHa83} $I(G;x)=I(G-v;x)+x\cdot I(G-N[v];x)$ holds
for every $v\in V(G)$.

\emph{(ii)} \cite{Gu92d} $I\left(  G\circ H;x\right)  =\left(  I\left(
H;x\right)  \right)  ^{n}\bullet I\left(  G;\frac{x}{I\left(  H;x\right)
}\right)  $, where $n=\left\vert V\left(  G\right)  \right\vert $.
\end{theorem}

A finite sequence of real numbers $(a_{0},a_{1},a_{2},...,a_{n})$ is said to be:

\begin{itemize}
\item \textit{unimodal} if there is some $k\in\{0,1,...,n\}$, such that
\[
a_{0}\leq...\leq a_{k-1}\leq a_{k}\geq a_{k+1}\geq...\geq a_{n};
\]

\item \textit{log-concave} if $a_{i}^{2}\geq a_{i-1}\cdot a_{i+1}$ for
$i\in\{1,2,...,n-1\}$.

\item \textit{symmetric} (or \textit{palindromic}) if $a_{i}=a_{n-i}%
,i=0,1,...,\left\lfloor n/2\right\rfloor $.
\end{itemize}

It is known that every log-concave sequence of positive numbers is also unimodal.

A polynomial is called \textit{unimodal (log-concave, symmetric)} if the
sequence of its coefficients is unimodal (log-concave, symmetric,
respectively). For instance, the independence polynomial

\begin{itemize}
\item $I(K_{42}+3K_{7};x)=1+63x+147x^{2}+343x^{3}$ is log-concave;

\item $I(K_{43}+3K_{7};x)=1+64x+147x^{2}+343x^{3}$ is unimodal, but
non-log-concave, because $147\cdot147-64\cdot343=-343<0$;

\item $I(K_{127}+3K_{7};x)=1+148x+\mathbf{147}x^{2}+343x^{3}$ is non-unimodal;

\item $I(K_{18}+3K_{3}+4K_{1};x)=1+\allowbreak31x+33x^{2}+31x^{3}+x^{4}$ is
symmetric and log-concave;

\item $I(K_{52}+3K_{4}+4K_{1};x)=1+68x+\mathbf{54}x^{2}+68x^{3}+x^{4}$ is
symmetric and non-unimodal.
\end{itemize}

It is easy to see that if $\alpha(G)\leq3$ and $I(G;x)$ is symmetric, then it
is also log-concave.

For other examples, see \cite{AlMaScEr87}, \cite{LeMa03b}, \cite{LeMa03c},
\cite{LeMa04c} and \cite{LeMa04b}. Alavi, Malde, Schwenk and Erd\"{o}s proved
that for any permutation $\pi$ of $\{1,2,...,\alpha\}$ there is a graph $G$
with $\alpha(G)=\alpha$ such that $s_{\pi(1)}<s_{\pi(2)}<...<s_{\pi(\alpha)}$
\cite{AlMaScEr87}.

In this paper we show that every graph $H$ derived from the graph $G$ by
Stevanovi\'{c}'s rules \cite{St98} gives rise to the following decomposition
\[
I\left(  G\circ2K_{1};x\right)  =\left(  1+x\right)  ^{k}\cdot I\left(
H;x\right)  ,
\]
for every non-negative integer $k\leq\mu\left(  G\right)  $.

\section{Preliminaries}

The symmetry of the matching polynomial and the characteristic polynomial of a
graph were examined in \cite{Kennedy}, while for the independence polynomial
we quote \cite{Gu93}, \cite{St98}, and \cite{BS2010}. Recall from
\cite{Kennedy} that $G$ is called a \textit{equible graph }if $G=H\circ K_{1}$
for some graph $H$. Both matching polynomials and characteristic polynomials
of equible graphs are symmetric \cite{Kennedy}. Nevertheless, there are
non-equible graphs whose matching polynomials and characteristic polynomials
are symmetric.

It is worth mentioning that one can produce graphs with symmetric independence
polynomials in different ways. We summarize some of them in the sequel.

\subsection{\textbf{Gutman's construction} \cite{Gu92c}}

For integers $p>1$, $q>1$, let $J_{p,q}$ be the graph built in the following
manner \cite{Gu92c}. Start with three complete graphs $K_{1}$, $K_{p}$ and
$K_{q}$ whose vertex sets are disjoint. Connect the vertex of $K_{1}$ with
$p-1$ vertices of $K_{p}$ and with $q-1$ vertices of $K_{q}$.
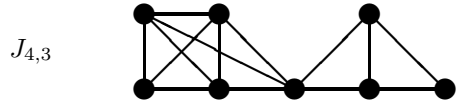
\begin{figure}[h]
\setlength{\unitlength}{1cm}\begin{picture}(5,1.3)\thicklines
\multiput(5,0)(1,0){5}{\circle*{0.29}}
\multiput(5,1)(1,0){2}{\circle*{0.29}}
\put(8,1){\circle*{0.29}}
\put(5,0){\line(1,0){4}}
\put(5,0){\line(0,1){1}}
\put(5,0){\line(1,1){1}}
\put(5,1){\line(1,0){1}}
\put(5,1){\line(1,0){1}}
\put(5,1){\line(1,-1){1}}
\put(6,0){\line(0,1){1}}
\put(7,0){\line(-1,1){1}}
\put(7,0){\line(-2,1){2}}
\put(7,0){\line(1,1){1}}
\put(8,0){\line(0,1){1}}
\put(8,1){\line(1,-1){1}}
\put(3.5,0.5){\makebox(0,0){$J_{4,3}$}}
\end{picture}
\caption{$I\left(  J_{4,3};x\right)  =1+8x+14x^{2}+x^{3}$ and $I\left(
J_{4,3}+K_{6};x\right)  =1+14x+14x^{2}+x^{3}$.}%
\end{figure}The graph thus obtained has a unique maximum independent set of
size three, and its independence polynomial is equal to
\[
I\left(  J_{p,q};x\right)  =1+(p+q+1)x+(pq+2)x^{2}+x^{3}.
\]
Hence the independence polynomial of $G=J_{p,q}+K_{pq-p-q+1}$ is
\[
I\left(  G;x\right)  =I\left(  J_{p,q};x\right)  +I\left(  K_{pq-p-q+1}%
;x\right)  -1=1+\left(  2+pq\right)  x+\left(  2+pq\right)  x^{2}+x^{3},
\]
which is clearly symmetric and log-concave.

\subsection{\textbf{Bahls and Salazar's construction} \cite{BS2010}}

The $K_{t}$-path of length $k\geq1$ is the graph $P(t,k)=(V,E)$ with
$V=\{v_{1},v_{2},...,v_{t+k-1}\}$ and $E=\left\{  v_{i}v_{i+j}:1\leq i\leq
t+k-2,1\leq j\leq\min\{t-1,t+k-i-1\}\right\}  $. Such a graph consists of $k $
copies of $K_{t}$, each glued to the previous one by identifying certain
prescribed subgraphs isomorphic to $K_{t-1}$. Let $d\geq0$ be an integer. The
$d$-augmented $K_{t}$ path $P(t,k,d)$ is defined by introducing new vertices
$\{u_{_{i},1},u_{_{i},2},...,u_{i,d}\}_{i=0}^{t+k-2}$ and edges $\left\{
v_{i}u_{i,j},v_{i+1}u_{i,j}:j=1,...,d\right\}  _{i=1}^{t+k-2}\cup\left\{
v_{1},u_{0,j}:j=1,...,d\right\}  $. Let $G=(V,E)$ and $U\subseteq V$ be a
subset of its vertices. Let $v\notin V$ and define the \textit{cone} of $G$ on
$U$ with vertex $v$, denoted $G^{\ast}(U,v)=\left(  G,U\right)  +K_{1}$, where
$K_{1}=\left(  \left\{  v\right\}  ,\emptyset\right)  $. Given $G$ and $U$ and
a graph $H$, we write $H+\left(  G,U\right)  $ instead of $\left(  H,V\left(
H\right)  \right)  +\left(  G,U\right)  $.

\begin{theorem}
\cite{BS2010} Let $t\geq2,k\geq1$, and $d\geq0$ be integers, and let $G=(V,E)$
be a graph with $U\subseteq V$ a distinguished subset of vertices. Suppose
that each of the graphs $G$, $G-U$, and $\left(  G,U\right)  +K_{1}$ have
symmetric and unimodal independence polynomials, and that $\deg(I(G;x))=\deg
(I(\left(  G,U\right)  +K_{1};x))=\deg(I(G-U;x))+2$. Then the independence
polynomial of the graph $P(t,k,d)+(G,U)$ is symmetric and unimodal.
\end{theorem}

\subsection{\textbf{Stevanovi\'{c}'s constructions} \cite{St98}}

Taking into account that $s_{0}=1$ and $s_{1}=\left\vert V(G)\right\vert =n$,
it follows that if $I(G;x)$ is symmetric, then $s_{0}=s_{\alpha}$ and
$s_{1}=s_{\alpha-1}$, i.e., $G$ has only one maximum independent set, say $S$,
and $n-\alpha(G)$ independent sets, of size $\alpha(G)-1$, that are not
subsets of $S$.

\begin{theorem}
\label{th1}\cite{St98} If there is an independent set $S$ in $G$ such that
$\left\vert N(A)\cap S\right\vert =2\left\vert A\right\vert $ holds for every
independent set $A\subseteq V\left(  G\right)  -S$, then $I(G;x)$ is symmetric.
\end{theorem}

The following result is a consequence of Theorem \ref{th1}.

\begin{corollary}
\label{cor1}\cite{St98} \emph{(i)} If $\alpha(G)=\alpha,s_{\alpha}%
=1,s_{\alpha-1}=\left\vert V(G)\right\vert $, and for the unique stability
system $S$ of $G$ it is true that $\left\vert N(v)\cap S\right\vert =2$ for
each $v\in V(G)-S$, then $I(G;x)$ is symmetric.

\emph{(ii)} If $G$ is a claw-free graph with $\alpha(G)=\alpha,s_{\alpha
}=1,s_{\alpha-1}=\left\vert V(G)\right\vert $, then $I(G;x)$ is symmetric.
\end{corollary}

Corollary \ref{cor1} gives three different ways to construct graphs having
symmetric independence polynomials \cite{St98}.

\begin{itemize}
\item \textbf{Rule 1.} For a given graph $G$, define a new graph $H$ as:
$H=G\circ2K_{1}$. \begin{figure}[h]
\setlength{\unitlength}{1cm}\begin{picture}(5,2.2)\thicklines
\multiput(1,0.5)(1,0){5}{\circle*{0.29}}
\put(2,1.5){\circle*{0.29}}
\put(1,0.5){\line(1,0){4}}
\put(2,0.5){\line(0,1){1}}
\put(2,1.5){\line(1,-1){1}}
\put(0.3,1){\makebox(0,0){$G$}}
\multiput(7.5,0)(1,0){5}{\circle*{0.29}}
\multiput(6.5,1)(1,0){7}{\circle*{0.29}}
\multiput(6.5,2)(1,0){6}{\circle*{0.29}}
\put(6.5,1){\line(1,0){6}}
\put(6.5,2){\line(1,-1){1}}
\put(7.5,2){\line(1,0){2}}
\put(7.5,0){\line(1,1){1}}
\put(8.5,0){\line(0,1){2}}
\put(8.5,2){\line(1,-1){2}}
\put(9.5,0){\line(0,1){1}}
\put(10.5,1){\line(0,1){1}}
\put(10.5,1){\line(1,1){1}}
\put(11.5,0){\line(0,1){1}}
\put(5.8,1){\makebox(0,0){$H_{1}$}}
\end{picture}
\caption{$G$ and $H_{1}=G\circ2K_{1}$.}%
\label{fig22}%
\end{figure}
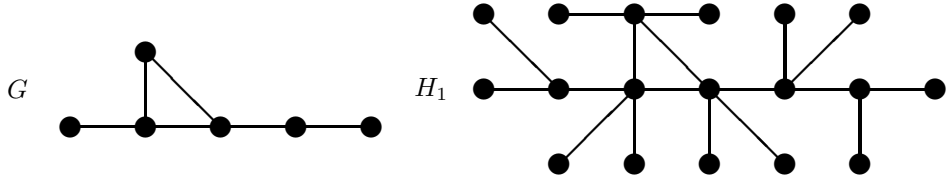
\end{itemize}

For an example, see the graphs in Figure \ref{fig22}: $I(G;x)=1+6x+9x^{2}%
+3x^{3}$, while%
\begin{align*}
I(H_{1};x)  &  =\left(  1+x\right)  ^{6}\left(  1+12x+48x^{2}+77x^{3}%
+48x^{4}+12x^{5}+x^{6}\right)  =\\
&  =1+18x+135x^{2}+565x^{3}+1485x^{4}+2601x^{5}+3126x^{6}+\\
&  +2601x^{7}+1485x^{8}+565x^{9}+135x^{10}+18x^{11}+x^{12}.
\end{align*}

\begin{itemize}
\item A \textit{cycle cover} of a graph $G$ is a spanning graph of $G$, each
connected component of which is a vertex (which we call a
\textit{vertex-cycle}), an edge (which we call an \textit{edge-cycle}), or a
proper cycle. Let $\Gamma$ be a cycle cover of $G$.

\textbf{Rule 2. }Construct a new graph $H$ from $G$, denoted by $H=\Gamma
\{G\}$, as follows: if $C\in\Gamma$ is

\emph{(i)} a vertex-cycle, say $v$, then add two vertices and join them to $v
$;

\emph{(ii)} an edge-cycle, say $uv$, then add two vertices and join them to
both $u$ and $v$;

\emph{(iii)} a proper cycle, with
\[
V(C)=\{v_{i}:1\leq i\leq s\},E(C)=\{v_{i}v_{i+1}:1\leq i\leq s-1\}\cup
\{v_{1}v_{s}\},
\]
then add $s$ vertices, say $\{w_{i}:1\leq i\leq s\}$ and each of them is
joined to two consecutive vertices on $C$, as follows: $w_{1}$ is joined to
$v_{s},v_{1}$, then $w_{2}$ is joined to $v_{1},v_{2}$, further $w_{3}$ is
joined to $v_{2},v_{3}$, etc.

Figure \ref{fig33} contains an example, namely, $I(G;x)=1+6x+9x^{2}+3x^{3}$,
while%
\begin{align*}
I(H_{2};x)  &  =1+13x+60x^{2}+125x^{3}+125x^{4}+60x^{5}+13x^{6}+x^{7}=\\
&  =\left(  1+x\right)  \left(  \allowbreak1+12x+48x^{2}+77x^{3}%
+48x^{4}+12x^{5}+x^{6}\right)  .
\end{align*}
\begin{figure}[h]
\setlength{\unitlength}{1cm}\begin{picture}(5,2)\thicklines
\multiput(2,0.5)(1,0){5}{\circle*{0.29}}
\put(3,1.5){\circle*{0.29}}
\put(2,0.5){\line(1,0){4}}
\put(3,0.5){\line(0,1){1}}
\put(3,1.5){\line(1,-1){1}}
\put(2,0.1){\makebox(0,0){$x$}}
\put(3,0.1){\makebox(0,0){$b$}}
\put(4,0.1){\makebox(0,0){$c$}}
\put(5,0.1){\makebox(0,0){$y$}}
\put(6,0.1){\makebox(0,0){$z$}}
\put(2.7,1.5){\makebox(0,0){$a$}}
\put(1.2,1){\makebox(0,0){$G$}}
\multiput(8,0)(1,0){4}{\circle*{0.29}}
\multiput(8,1)(1,0){5}{\circle*{0.29}}
\multiput(8,2)(1,0){4}{\circle*{0.29}}
\put(8,1){\line(1,0){4}}
\put(8,1){\line(1,-1){1}}
\put(8,0){\line(0,1){1}}
\put(8,2){\line(1,-1){2}}
\put(8,2){\line(1,0){2}}
\put(9,2){\line(1,-1){1}}
\put(9,1){\line(0,1){1}}
\put(10,0){\line(0,1){2}}
\put(11,0){\line(0,1){2}}
\put(11,0){\line(1,1){1}}
\put(11,2){\line(1,-1){1}}
\put(7.2,1){\makebox(0,0){$H_{2}$}}
\end{picture}
\caption{$G$ and $H_{2}=\Gamma\left(  G\right)  $, where $\Gamma=\left\{
\left\{  x\right\}  ,\left\{  a,b,c\right\}  ,\left\{  y,z\right\}  \right\}
$.}%
\label{fig33}%
\end{figure}
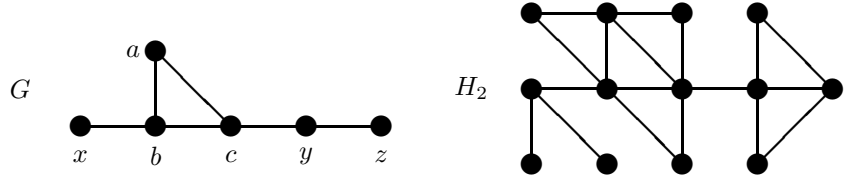

\item A \textit{clique cover} of a graph $G$ is a spanning graph of $G$, each
connected component of which is a clique. Let $\Phi$ be a clique cover of $G$.

\textbf{Rule 3. }Construct a new graph $H$ from $G$, denoted by $H=\Phi\{G\}
$, as follows: for each $Q\in\Phi$, add two non-adjacent vertices and join
them to all the vertices of $Q$.
\end{itemize}

Figure \ref{fig444} contains an example, namely, $I(G;x)=1+6x+9x^{2}+3x^{3}$,
while%
\[
I(H_{3};x)=1+12x+48x^{2}+77x^{3}+48x^{4}+12x^{5}+x^{6}.
\]
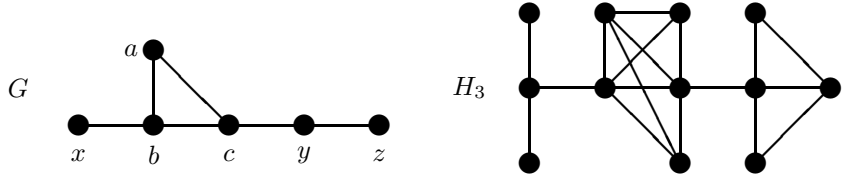
\begin{figure}[h]
\setlength{\unitlength}{1cm}\begin{picture}(5,2.3)\thicklines
\multiput(2,0.5)(1,0){5}{\circle*{0.29}}
\put(3,1.5){\circle*{0.29}}
\put(2,0.5){\line(1,0){4}}
\put(3,0.5){\line(0,1){1}}
\put(3,1.5){\line(1,-1){1}}
\put(2,0.1){\makebox(0,0){$x$}}
\put(3,0.1){\makebox(0,0){$b$}}
\put(4,0.1){\makebox(0,0){$c$}}
\put(5,0.1){\makebox(0,0){$y$}}
\put(6,0.1){\makebox(0,0){$z$}}
\put(2.7,1.5){\makebox(0,0){$a$}}
\put(1.2,1){\makebox(0,0){$G$}}
\multiput(8,0)(2,0){2}{\circle*{0.29}}
\multiput(8,1)(1,0){5}{\circle*{0.29}}
\multiput(8,2)(1,0){4}{\circle*{0.29}}
\put(11,0){\circle*{0.29}}
\put(8,0){\line(0,1){2}}
\put(8,1){\line(1,0){4}}
\put(9,1){\line(0,1){1}}
\put(9,1){\line(1,-1){1}}
\put(9,1){\line(1,1){1}}
\put(9,2){\line(1,0){1}}
\put(9,2){\line(1,-1){1}}
\put(9,2){\line(1,-2){1}}
\put(10,0){\line(0,1){2}}
\put(11,0){\line(0,1){2}}
\put(11,0){\line(1,1){1}}
\put(11,2){\line(1,-1){1}}
\put(7.2,1){\makebox(0,0){$H_{3}$}}
\end{picture}
\caption{$G$ and $H_{3}=\Phi\left(  G\right)  $, where $\Phi=\left\{  \left\{
x\right\}  ,\left\{  a,b,c\right\}  ,\left\{  y,z\right\}  \right\}  $.}%
\label{fig444}%
\end{figure}

\begin{theorem}
\label{th4}\cite{St98} Let $H$ be the graph obtained from a graph $G$
according to one of the \textbf{Rules 1,2} or \textbf{3}. Then $H$ has a
symmetric independence polynomial.
\end{theorem}

Let us remark that $I(H_{1};x)=\left(  1+x\right)  ^{6}\cdot I(H_{3};x)$ and
$I(H_{2};x)=\left(  1+x\right)  \cdot I(H_{3};x)$, where $H_{1},H_{2}$ and
$H_{3}$ are depicted in Figures \ref{fig22}, \ref{fig33}, and \ref{fig444}, respectively.

\subsection{Inequalities and equalities following from Theorem \ref{th4}}

\begin{proposition}
\cite{LevMan2008} Let $G=H\circ2K_{1}$ be with $\alpha(G)=\alpha$, and
$\left(  s_{k}\right)  $ be the coefficients of $I(G;x)$.\ Then $I(G;x)$ is
symmetric, and
\begin{align*}
s_{0}  &  \leq s_{1}\leq...\leq s_{p}\text{ \textit{for} }p=\left\lfloor
(2\alpha+2)/5\right\rfloor \text{, while }\\
s_{t}  &  \geq...\geq s_{\alpha-1}\geq s_{\alpha}\text{ for}\ t=\left\lceil
(3\alpha-2)/5\right\rceil .
\end{align*}

\end{proposition}

\begin{theorem}
\cite{LevMan2008} Let $H$ be a graph of order $n\geq2$, $\Gamma$ be a cycle
cover of $H$ that contains no vertex-cycles, $G$ be obtained by \textbf{Rule
2}, and $\alpha(G)=\alpha$. Then $I(G;x)$ is symmetric and its coefficients
$(s_{k})$ satisfy the subsequent inequalities:%
\begin{align*}
s_{0}  &  \leq s_{1}\leq...\leq s_{p}\text{, for}~\ p=\left\lfloor
(\alpha+1)/3\right\rfloor \text{, and }\\
s_{q}  &  \geq...\geq s_{\alpha-1}\geq s_{\alpha}\text{, for}\ ~q=\left\lceil
(2\alpha-1)/3\right\rceil .
\end{align*}

\end{theorem}

Let $H_{n},n\geq1$, be the graphs obtained according to \textbf{Rule 3} from
$P_{n}$, as one can see in Figure \ref{Fig44}.

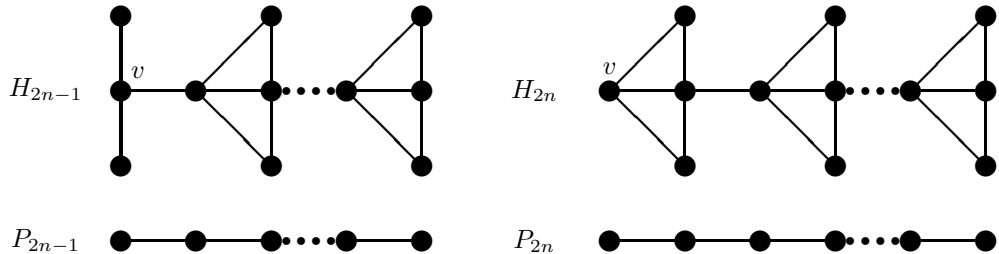
\begin{figure}[h]
\setlength{\unitlength}{1cm}\begin{picture}(5,3.2)\thicklines
\multiput(1.5,0)(1,0){5}{\circle*{0.29}}
\multiput(1.5,1)(2,0){3}{\circle*{0.29}}
\multiput(1.5,2)(1,0){5}{\circle*{0.29}}
\multiput(1.5,3)(2,0){3}{\circle*{0.29}}
\put(1.5,0){\line(1,0){2}}
\put(1.5,1){\line(0,1){2}}
\put(1.5,2){\line(1,0){2}}
\put(1.5,1){\line(0,1){2}}
\put(4.5,0){\line(1,0){1}}
\put(4.5,2){\line(1,0){1}}
\multiput(3.5,0)(0.2,0){5}{\circle*{0.10}}
\multiput(3.5,2)(0.2,0){5}{\circle*{0.10}}
\put(2.5,2){\line(1,1){1}}
\put(2.5,2){\line(1,-1){1}}
\put(3.5,1){\line(0,1){2}}
\put(4.5,2){\line(1,1){1}}
\put(4.5,2){\line(1,-1){1}}
\put(5.5,1){\line(0,1){2}}
\put(1.73,2.25){\makebox(0,0){$v$}}
\put(0.5,2){\makebox(0,0){$H_{2n-1}$}}
\put(0.5,0){\makebox(0,0){$P_{2n-1}$}}
\multiput(8,0)(1,0){6}{\circle*{0.29}}
\multiput(9,1)(2,0){3}{\circle*{0.29}}
\multiput(8,2)(1,0){6}{\circle*{0.29}}
\multiput(9,3)(2,0){3}{\circle*{0.29}}
\put(8,0){\line(1,0){3}}
\put(8,2){\line(1,1){1}}
\put(8,2){\line(1,-1){1}}
\put(9,1){\line(0,1){2}}
\put(8,2){\line(1,0){3}}
\multiput(11,0)(0.2,0){5}{\circle*{0.10}}
\multiput(11,2)(0.2,0){5}{\circle*{0.10}}
\put(12,0){\line(1,0){1}}
\put(12,2){\line(1,0){1}}
\put(10,2){\line(1,1){1}}
\put(10,2){\line(1,-1){1}}
\put(11,1){\line(0,1){2}}
\put(12,2){\line(1,1){1}}
\put(12,2){\line(1,-1){1}}
\put(13,1){\line(0,1){2}}
\put(8,2.3){\makebox(0,0){$v$}}
\put(7,2){\makebox(0,0){$H_{2n}$}}
\put(7,0){\makebox(0,0){$P_{2n}$}}
\end{picture}
\caption{$P_{n}$ and $H_{n}=\Omega\{P_{n}\}$.}%
\label{Fig44}%
\end{figure}

\begin{theorem}
\cite{LevMan2007} If $J_{n}(x)=I(H_{n};x),n\geq0$, then

\emph{(i)} $J_{0}(x)=1,J_{1}(x)=1+3x+x^{2}$ and $J_{n},n\geq2$, satisfies the
following recursive relations:
\begin{align*}
J_{2n}(x)  &  =J_{2n-1}(x)+x\cdot J_{2n-2}(x),\quad n\geq1,\\
J_{2n-1}(x)  &  =(1+x)^{2}\cdot J_{2n-2}(x)+x\cdot J_{2n-3}(x),\quad n\geq2;
\end{align*}

\emph{(ii)} $J_{n}$ is both symmetric and unimodal.
\end{theorem}

It was conjectured in \cite{LevMan2007} that $I(H_{n};x)$ is log-concave and
has only real roots. This conjecture has been resolved as follows.

\begin{theorem}
\cite{Wang} Let $n\geq1$. Then

\emph{(i)} the independence polynomial of $H_{n}$ is%
\[
I(H_{n};x)=\prod\limits_{s=1}^{\left\lfloor \left(  n+1\right)
/2\right\rfloor }\left(  1+4x+x^{2}+2x\cdot\cos\frac{2s\pi}{n+2}\right)  ;
\]

\emph{(ii)} $I(H_{n};x)$ has only real zeros, and, therefore, it is
log-concave and unimodal.
\end{theorem}

\section{Results}

The following lemma goes from the well-known fact that the polynomial $P(x)$
is symmetric if and only if it equals its reciprocal, i.e.,
\begin{equation}
P(x)=x^{\deg(P)}P\left(  \frac{1}{x}\right)  . \tag{*}\label{idenity}%
\end{equation}

\begin{lemma}
\label{lem1}Let $f\left(  x\right)  $, $g\left(  x\right)  $ and $h\left(
x\right)  $ be polynomials satisfying $f\left(  x\right)  =$ $g\left(
x\right)  \cdot$ $h\left(  x\right)  $. If any two of them are symmetric, then
the third is symmetric as well.
\end{lemma}

For $H=2K_{1}$, Theorem \ref{th3} gives
\[
I\left(  G\circ2K_{1};x\right)  =\left(  1+x\right)  ^{2n}\cdot I\left(
G;\frac{x}{\left(  1+x\right)  ^{2}}\right)  .
\]
Since $\frac{x}{\left(  1+x\right)  ^{2}}=\frac{\frac{1}{x}}{\left(
1+\frac{1}{x}\right)  ^{2}}$ and $\deg\left(  I\left(  G\circ2K_{1};x\right)
\right)  =2n$, one can easily see that the polynomial $I\left(  G\circ
2K_{1};x\right)  $ satisfies the identity (\ref{idenity}). Thus we conclude
with the following.

\begin{theorem}
\cite{St98} For every graph $G$, the polynomial $I\left(  G\circ
2K_{1};x\right)  $ is symmetric.
\end{theorem}

\subsection{Clique covers}

\begin{lemma}
If $A$ is a clique in a graph $G$, then for every graph $H$%
\[
I((G,A)\circ H;x)=I\left(  H;x\right)  ^{\left\vert A\right\vert -1}\cdot
I((G,A)+H;x).
\]

\end{lemma}

\begin{proof}
Let $G_{1}=(G,A)\circ H$ and $G_{2}=\left(  (G,A)+H\right)  \cup\left(
(\left\vert A\right\vert -1)H\right)  $.

For $S\in$ \textrm{Ind}$(G)$, let denote the following families of independent
sets:
\begin{align*}
\Omega_{S}^{G_{1}}  &  =\{S\cup W:W\subseteq V(G_{1}-G),S\cup W\in
\text{\textrm{Ind}}(G_{1})\},\\
\Omega_{S}^{G_{2}}  &  =\{S\cup W:W\subseteq V(G_{2}-G),S\cup W\in
\mathrm{Ind}(G_{2})\}.
\end{align*}
Since $A$ is a clique, it follows that $\left\vert S\cap A\right\vert \leq1$.

\emph{Case 1}. $S\cap A=\varnothing$.

In this case $S\cup W\in\Omega_{S}^{G_{1}}$ if and only if $S\cup W\in
\Omega_{S}^{G_{2}}$. Hence, for each size $m\geq\left\vert S\right\vert $, we
get that
\[
\left\vert \{S\cup W\in\Omega_{S}^{G_{1}}:\left\vert S\cup W\right\vert
=m\}\right\vert =\left\vert \{S\cup W\in\Omega_{S}^{G_{2}}:\left\vert S\cup
W\right\vert =m\}\right\vert .
\]

\emph{Case 2}. $S\cap A=\{a\}$.

Now, every $S\cup W\in\Omega_{S}^{G_{1}}$ has $W\cap V(H)=\varnothing$ for
exactly one $H$, namely, the graph $H$ whose vertices are joined to $a$.
Hence, $W$ may contain vertices only from $\left(  |A|-1\right)  H$.

On the other hand, each $S\cup W\in\Omega_{S}^{G_{2}}$ has $W\cap
V(H)=\varnothing$ for the unique $H$\ appearing in $(G,A)+H$. Therefore, $W$
may contain vertices only from $\left(  |A|-1\right)  H$.

Hence, for each positive integer $m\geq\left\vert S\right\vert $, we obtain
that
\[
\left\vert \{S\cup W\in\Omega_{S}^{G_{1}}:\left\vert S\cup W\right\vert
=m\}\right\vert =\left\vert \{S\cup W\in\Omega_{S}^{G_{2}}:\left\vert S\cup
W\right\vert =m\}\right\vert .
\]

Consequently, one may infer that for each size, the two graphs, $G_{1}$ and
$G_{2}$, have the same number of independent sets, in other words,
$I(G_{1};x)=I(G_{2};x)$.

Since $G_{2}=\left(  (G,A)+H\right)  \cup\left(  (\left\vert A\right\vert
-1)H\right)  $ has $\left\vert A\right\vert -1$ disjoint components identical
to $H$, it follows that $I(G_{2};x)=I\left(  H;x\right)  ^{\left\vert
A\right\vert -1}\cdot I((G,A)+H;x)$.
\end{proof}

\begin{corollary}
\label{cor2}If $A$ is a clique in a graph $G$, then%
\[
I((G,A)\circ2K_{1};x)=\left(  1+x\right)  ^{2\left\vert A\right\vert -2}\cdot
I((G,A)+2K_{1};x).
\]

\end{corollary}

\begin{theorem}
\label{th2}If $G$ is a graph of order $n$ and $\Phi$\ is a clique cover, then%
\[
I(G\circ2K_{1};x)=\left(  1+x\right)  ^{2n-2\left\vert \Phi\right\vert }\cdot
I(\Phi(G);x).
\]

\end{theorem}

\begin{proof}
Let $\Phi=\left\{  A_{1},A_{2},...,A_{q}\right\}  $. According to Corollary
\ref{cor2}, each

\emph{(a)} vertex-clique of $\Phi$ yields $\left(  1+x\right)  ^{2-2}=1$ as a
factor of $I(G\circ2K_{1};x)$, since a vertex defines a clique of size $1$;

\emph{(b)} edge-clique of $\Phi$ yields $\left(  1+x\right)  ^{2}$ as a factor
of $I(G\circ2K_{1};x)$, since an edge defines a clique of size $2$;
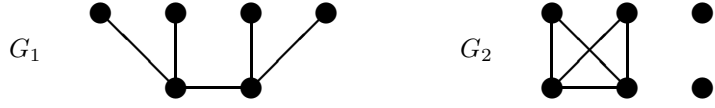
\begin{figure}[h]
\setlength{\unitlength}{1cm}\begin{picture}(5,1.3)\thicklines
\multiput(4,0)(1,0){2}{\circle*{0.29}}
\multiput(3,1)(1,0){4}{\circle*{0.29}}
\put(4,0){\line(1,0){1}}
\put(4,0){\line(0,1){1}}
\put(4,0){\line(-1,1){1}}
\put(5,0){\line(0,1){1}}
\put(5,0){\line(1,1){1}}
\put(2,0.5){\makebox(0,0){$G_{1}$}}
\multiput(9,0)(1,0){3}{\circle*{0.29}}
\multiput(9,1)(1,0){3}{\circle*{0.29}}
\put(9,0){\line(1,1){1}}
\put(9,0){\line(0,1){1}}
\put(9,0){\line(1,0){1}}
\put(9,1){\line(1,-1){1}}
\put(10,0){\line(0,1){1}}
\put(8,0.5){\makebox(0,0){$G_{2}$}}
\end{picture}
\caption{$G_{1}=K_{2}\circ2K_{1}$, $I\left(  G_{1};x\right)  =\left(
1+x\right)  ^{2}\cdot I\left(  \Phi\left(  K_{2}\right)  ;x\right)  =\left(
1+x\right)  ^{2}\cdot\left(  1+4x+x^{2}\right)  $. }%
\end{figure}

\emph{(c)} clique $A_{j}\in\Phi,\left\vert A_{j}\right\vert \geq3$, produces
$\left(  1+x\right)  ^{2\left\vert A_{j}\right\vert -2}$ as a factor of
$I(G\circ2K_{1};x)$.

Since the cliques of $\Phi$ are pairwise vertex disjoint, one can apply
Corollary \ref{cor2} to all the $q$ cliques one by one. \begin{figure}[h]
\setlength{\unitlength}{1cm}\begin{picture}(5,3.1)\thicklines
\multiput(4,0)(1,0){2}{\circle*{0.29}}
\multiput(3,1)(1,0){4}{\circle*{0.29}}
\multiput(3,2)(1,0){4}{\circle*{0.29}}
\multiput(4,3)(1,0){2}{\circle*{0.29}}
\put(4,0){\line(0,1){3}}
\put(4,1){\line(1,1){1}}
\put(4,2){\line(1,-1){1}}
\put(5,0){\line(0,1){3}}
\put(3,1){\line(1,0){3}}
\put(3,2){\line(1,0){3}}
\put(1.8,1.5){\makebox(0,0){$G_{1}$}}
\multiput(9,0.5)(1,0){4}{\circle*{0.29}}
\multiput(9,1.5)(1,0){4}{\circle*{0.29}}
\multiput(9,2.5)(1,0){4}{\circle*{0.29}}
\put(9,0.5){\line(1,0){3}}
\put(9,0.5){\line(1,1){1}}
\put(9,0.5){\line(2,1){2}}
\put(10,0.5){\line(0,1){1}}
\put(10,0.5){\line(1,1){1}}
\put(10,1.5){\line(1,0){1}}
\put(10,1.5){\line(1,-1){1}}
\put(10,1.5){\line(2,-1){2}}
\put(11,0.5){\line(0,1){1}}
\put(11,1.5){\line(1,-1){1}}
\qbezier(9,0.5)(10,-0.3)(11,0.5)
\qbezier(10,0.5)(11,-0.3)(12,0.5)
\put(7.8,1.5){\makebox(0,0){$G_{2}$}}
\end{picture}
\caption{$G_{1}=K_{4}\circ2K_{1}$, $G_{2}=6K_{1}\cup\Phi\left(  K_{4}\right)
$ and $I\left(  G_{1};x\right)  =\left(  1+x\right)  ^{6}\cdot I\left(
\Phi\left(  K_{4}\right)  ;x\right)  $.}%
\label{fig1212}%
\end{figure}
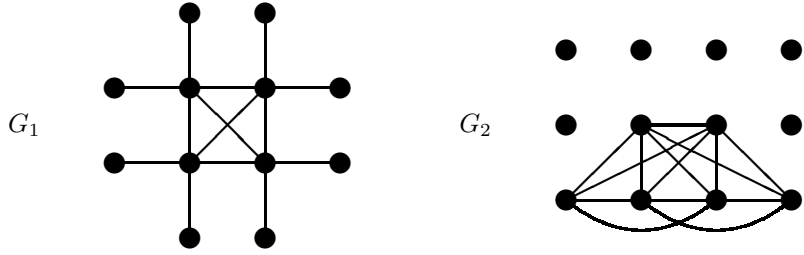

Using Corollary \ref{cor2} and the fact that $A_{1}\cap A_{2}=\emptyset$, we
have%
\begin{gather*}
I((G,A_{1}\cup A_{2})\circ2K_{1};x)=I(\left(  \left(  (G,A_{1})\circ
2K_{1}\right)  ,A_{2}\right)  \circ2K_{1};x)=\\
=\left(  1+x\right)  ^{2\left\vert A_{2}\right\vert -2}\cdot I(\left(  \left(
(G,A_{1})\circ2K_{1}\right)  ,A_{2}\right)  +2K_{1};x)\\
=\left(  1+x\right)  ^{2\left\vert A_{2}\right\vert -2}\cdot I(\left(  \left(
(G,A_{2})+2K_{1}\right)  ,A_{1}\right)  \circ2K_{1};x)\\
=\left(  1+x\right)  ^{2\left(  \left\vert A_{1}\right\vert +\left\vert
A_{2}\right\vert \right)  -2}\cdot I(\left(  \left(  (G,A_{2})+2K_{1}\right)
,A_{1}\right)  +2K_{1};x).
\end{gather*}

Repeating this process with $\left\{  A_{3},A_{4},...,A_{q}\right\}  $, and
taking into account that all the cliques of $\Phi$\ are pairwise disjoint, we
obtain%
\begin{gather*}
I((G\circ2K_{1};x)=I((G,A_{1}\cup A_{2}\cup...\cup A_{q})\circ2K_{1};x)=\\
=\left(  1+x\right)  ^{2\left(  \left\vert A_{1}\right\vert +\left\vert
A_{2}\right\vert +...+\left\vert A_{q}\right\vert \right)  -2q}\cdot I(\left(
(\left(  (G,A_{1})+2K_{1}\right)  ,A_{2}...),A_{q}\right)  +2K_{1};x)=\\
=\left(  1+x\right)  ^{2n-2\left\vert \Phi\right\vert }\cdot I(\Phi(G);x),
\end{gather*}
as required.
\end{proof}

Lemma \ref{lem1} and Theorem \ref{th2} imply the following.

\begin{corollary}
\cite{St98} For every clique cover $\Phi$ of a graph $G$, the polynomial
$I(\Phi(G);x)$ is symmetric.
\end{corollary}

\subsection{Cycle covers}

\begin{lemma}
If $C$ is a proper cycle in a graph $G$, then for every graph $H$%
\[
I((G,C)\circ2H;x)=I\left(  H;x\right)  ^{\left\vert C\right\vert }\cdot
I((G,C)\bigtriangleup H;x).
\]

\end{lemma}

\begin{proof}
Let $C=\left(  V\left(  C\right)  ,E\left(  C\right)  \right)  $,
$q=\left\vert V\left(  C\right)  \right\vert $, $G_{1}=(G,C)\circ2H$, and
$G_{2}=\left(  (G,C)\bigtriangleup H\right)  \cup\left(  qH\right)  $.

For an independent set $S\subset V(G)$, let us denote:%
\begin{align*}
\Omega_{S}^{G_{1}}  &  =\{S\cup W:W\subseteq V(G_{1})-V(G),S\cup
W\in\text{\textrm{Ind}}(G_{1})\}\text{ and }\\
\Omega_{S}^{G_{2}}  &  =\{S\cup W:W\subseteq V(G_{2})-V(G),S\cup
W\in\text{\textrm{Ind}}(G_{2})\}.
\end{align*}

\emph{Case 1}. $S\cap V(C)=\emptyset$.

In this case $S\cup W\in\Omega_{S}^{G_{1}}$ if an only if $S\cup W\in
\Omega_{S}^{G_{2}}$, since $W$ is an arbitrary independent set of $2qH$.
Hence, for each size $m\geq\left\vert S\right\vert $, we get that
\[
\left\vert \{S\cup W\in\Omega_{S}^{G_{1}}:\left\vert S\cup W\right\vert
=m\}\right\vert =\left\vert \{S\cup W\in\Omega_{S}^{G_{2}}:\left\vert S\cup
W\right\vert =m\}\right\vert .
\]

\emph{Case 2}. $S\cap V(C)\neq\emptyset$.

Then, we may assert that
\[
\left\vert \Omega_{S}^{G_{1}}\right\vert =\left\vert \{S\cup W:W\text{
\textit{is an independent set in} }2(q-\left\vert S\cap V(C)\right\vert
)H\}\right\vert =\left\vert \Omega_{S}^{G_{2}}\right\vert \text{,}%
\]
since $W$ has to avoid all the "$H$-neighbors" of the vertices in $S\cap
V(C)$, both in $G_{1}$ and $G_{2}$.

Hence, for each positive integer $m\geq\left\vert S\right\vert $, we get that
\[
\left\vert \{S\cup W\in\Omega_{S}^{G_{1}}:\left\vert S\cup W\right\vert
=m\}\right\vert =\left\vert \{S\cup W\in\Omega_{S}^{G_{2}}:\left\vert S\cup
W\right\vert =m\}\right\vert .
\]
Consequently, one may infer that for each size, the two graphs, $G_{1}$ and
$G_{2}$, have the same number of independent sets. In other words,
$I(G_{1};x)=I(G_{2};x)$.

Since $G_{2}$ has $\left\vert C\right\vert $ disjoint components identical to
$H$, it follows that $I(G_{2};x)=\left(  1+x\right)  ^{\left\vert C\right\vert
}\cdot I((G,C)\bigtriangleup H;x)$.
\end{proof}

\begin{corollary}
\label{cor3}If $C$ is a\ proper cycle in a graph $G$, then%
\[
I((G,C)\circ2K_{1};x)=\left(  1+x\right)  ^{\left\vert C\right\vert }\cdot
I((G,C)\bigtriangleup K_{1};x).
\]

\end{corollary}

\begin{theorem}
\label{th5}If $G$ is a graph of order $n$ and $\Gamma$\ is a cycle cover
containing $k$ vertex-cycles, then $I(G\circ2K_{1};x)$ satisfies%
\[
I(G\circ2K_{1};x)=\left(  1+x\right)  ^{n-k}\cdot I(\Gamma(G);x).
\]

\end{theorem}

\begin{proof}
According to Corollaries \ref{cor2} and \ref{cor3}, each

\emph{(a)} vertex-cycle of $\Gamma$ yields $\left(  1+x\right)  ^{2-2}=1$ as a
factor of $I(G\circ2K_{1};x)$, since a vertex defines a clique of size $1$;

\emph{(b)} edge-cycle of $\Gamma$ yields $\left(  1+x\right)  ^{2}$ as a
factor of $I(G\circ2K_{1};x)$, since an edge defines a clique of size $2$;
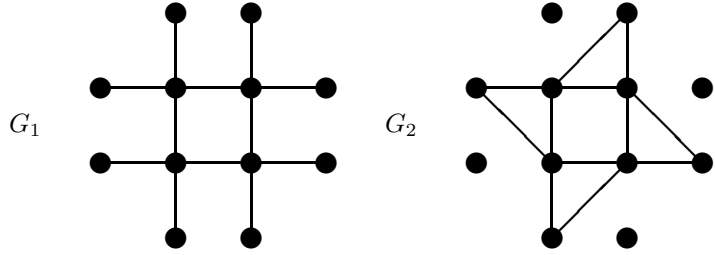
\begin{figure}[h]
\setlength{\unitlength}{1cm}\begin{picture}(5,3.1)\thicklines
\multiput(4,0)(1,0){2}{\circle*{0.29}}
\multiput(3,1)(1,0){4}{\circle*{0.29}}
\multiput(3,2)(1,0){4}{\circle*{0.29}}
\multiput(4,3)(1,0){2}{\circle*{0.29}}
\put(4,0){\line(0,1){3}}
\put(5,0){\line(0,1){3}}
\put(3,1){\line(1,0){3}}
\put(3,2){\line(1,0){3}}
\put(2,1.5){\makebox(0,0){$G_{1}$}}
\multiput(9,0)(1,0){2}{\circle*{0.29}}
\multiput(8,1)(1,0){4}{\circle*{0.29}}
\multiput(8,2)(1,0){4}{\circle*{0.29}}
\multiput(9,3)(1,0){2}{\circle*{0.29}}
\put(8,2){\line(1,0){2}}
\put(9,2){\line(1,1){1}}
\put(9,0){\line(1,1){1}}
\put(9,0){\line(0,1){2}}
\put(9,1){\line(1,0){2}}
\put(9,1){\line(-1,1){1}}
\put(10,2){\line(1,-1){1}}
\put(10,1){\line(0,1){2}}
\put(7,1.5){\makebox(0,0){$G_{2}$}}
\end{picture}
\caption{$G_{1}=C_{4}\circ2K_{1}$, $G_{2}=4K_{1}\cup\Gamma\left(
C_{4}\right)  $ and $I\left(  G_{1};x\right)  =\left(  1+x\right)  ^{4}\cdot
I\left(  \Gamma\left(  C_{4}\right)  ;x\right)  $}%
\label{fig121}%
\end{figure}

\emph{(c)} proper cycle $C\in\Gamma$ produces $\left(  1+x\right)
^{\left\vert C\right\vert }$ as a factor.

Let $\Gamma=\left\{  C_{j}:1\leq j\leq q\right\}  \cup\left\{  v_{i}:1\leq
i\leq k\right\}  $ be a cycle cover containing $k$ vertex-cycles, namely,
$\left\{  v_{i}:1\leq i\leq k\right\}  $.

Using Corollary \ref{cor3} and the fact that $C_{1}\cap C_{2}=\emptyset$, we
have%
\begin{gather*}
I((G,C_{1}\cup C_{2})\circ2K_{1};x)=I(\left(  \left(  (G,C_{1})\circ
2K_{1}\right)  ,C_{2}\right)  \circ2K_{1};x)=\\
=\left(  1+x\right)  ^{\left\vert C_{2}\right\vert }\cdot I(\left(  \left(
(G,C_{1})\circ2K_{1}\right)  ,C_{2}\right)  \bigtriangleup K_{1};x)\\
=\left(  1+x\right)  ^{\left\vert C_{2}\right\vert }\cdot I(\left(  \left(
(G,C_{2})\bigtriangleup K_{1}\right)  ,C_{1}\right)  \circ2K_{1};x)\\
=\left(  1+x\right)  ^{\left\vert C_{1}\right\vert +\left\vert C_{2}%
\right\vert }\cdot I(\left(  \left(  (G,C_{2})\bigtriangleup K_{1}\right)
,C_{1}\right)  \bigtriangleup K_{1};x).
\end{gather*}

Repeating this process with $\left\{  C_{3},C_{4},...,C_{q}\right\}  $, and
taking into account that all the cycles of $\Gamma$\ are pairwise vertex
disjoint, we obtain%
\begin{gather*}
I((G\circ2K_{1};x)=I((G,C_{1}\cup C_{2}\cup...\cup C_{q})\circ2K_{1};x)=\\
=\left(  1+x\right)  ^{\left\vert C_{1}\right\vert +\left\vert C_{2}%
\right\vert +...+\left\vert C_{q}\right\vert }\cdot I(\left(  (\left(
(G,C_{1})\bigtriangleup K_{1}\right)  ,C_{2}...),C_{q}\right)  \bigtriangleup
K_{1};x)=\\
=\left(  1+x\right)  ^{n-k}\cdot I(\Gamma(G);x),
\end{gather*}
as claimed.
\end{proof}

Lemma \ref{lem1} and Theorem \ref{th5} imply the following.

\begin{corollary}
\cite{St98} For every cycle cover $\Gamma$ of a graph $G$, the polynomial
$I(\Gamma(G);x)$ is symmetric.
\end{corollary}

\section{Conclusions\qquad}

In this paper we have given algebraic proofs for the assertions in Theorem
\ref{th4}, due to Stevanovi\'{c} \cite{St98}. In addition, we have showed that
for every clique cover $\Phi$, and every cycle cover $\Gamma$\ of a graph $G$,
the polynomial $I(G\circ2K_{1};x)$ is divisible both by $I(\Phi(G);x)$ and
$I(\Gamma(G);x)$.

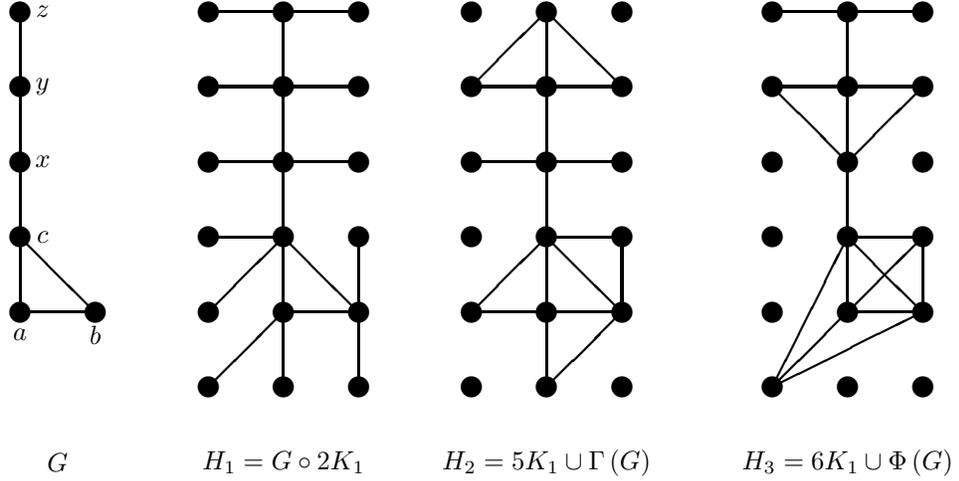
\begin{figure}[h]
\setlength{\unitlength}{1cm}\begin{picture}(5,6.3)\thicklines
\multiput(1,2)(0,1){5}{\circle*{0.29}}
\put(2,2){\circle*{0.29}}
\put(1,2){\line(0,1){4}}
\put(1,3){\line(1,-1){1}}
\put(1,2){\line(1,0){1}}
\put(1,1.7){\makebox(0,0){$a$}}
\put(2,1.7){\makebox(0,0){$b$}}
\put(1.3,3){\makebox(0,0){$c$}}
\put(1.3,4){\makebox(0,0){$x$}}
\put(1.3,5){\makebox(0,0){$y$}}
\put(1.3,6){\makebox(0,0){$z$}}
\put(1.5,0){\makebox(0,0){$G$}}
\multiput(3.5,1)(0,1){6}{\circle*{0.29}}
\multiput(4.5,1)(0,1){6}{\circle*{0.29}}
\multiput(5.5,1)(0,1){6}{\circle*{0.29}}
\put(4.5,1){\line(0,1){5}}
\put(4.5,2){\line(1,0){1}}
\put(4.5,3){\line(1,-1){1}}
\put(3.5,1){\line(1,1){1}}
\put(3.5,2){\line(1,1){1}}
\put(3.5,3){\line(1,0){1}}
\put(3.5,4){\line(1,0){2}}
\put(3.5,5){\line(1,0){2}}
\put(3.5,6){\line(1,0){2}}
\put(5.5,1){\line(0,1){2}}
\put(4.5,0){\makebox(0,0){$H_{1}=G\circ2K_{1}$}}
\multiput(7,1)(0,1){6}{\circle*{0.29}}
\multiput(8,1)(0,1){6}{\circle*{0.29}}
\multiput(9,1)(0,1){6}{\circle*{0.29}}
\put(8,1){\line(0,1){5}}
\put(8,1){\line(1,1){1}}
\put(7,2){\line(1,0){2}}
\put(7,2){\line(1,1){1}}
\put(8,3){\line(1,0){1}}
\put(7,4){\line(1,0){2}}
\put(7,5){\line(1,0){2}}
\put(7,5){\line(1,1){1}}
\put(8,6){\line(1,-1){1}}
\put(8,3){\line(1,-1){1}}
\put(9,2){\line(0,1){1}}
\put(8,0){\makebox(0,0){$H_{2}=5K_{1}\cup\Gamma\left(G\right)$}}
\multiput(11,1)(0,1){6}{\circle*{0.29}}
\multiput(12,1)(0,1){6}{\circle*{0.29}}
\multiput(13,1)(0,1){6}{\circle*{0.29}}
\put(11,1){\line(1,1){2}}
\put(11,1){\line(1,2){1}}
\put(11,1){\line(2,1){2}}
\put(11,5){\line(1,-1){1}}
\put(11,5){\line(1,0){2}}
\put(11,6){\line(1,0){2}}
\put(12,2){\line(1,0){1}}
\put(12,3){\line(1,0){1}}
\put(12,3){\line(1,-1){1}}
\put(12,4){\line(1,1){1}}
\put(12,2){\line(0,1){4}}
\put(13,2){\line(0,1){1}}
\put(12,0){\makebox(0,0){$H_{3}=6K_{1}\cup\Phi\left(G\right)$}}
\end{picture}
\caption{$G$ with $\Gamma\left(  G\right)  =\left\{  \left\{  y,z\right\}
,\left\{  x\right\}  ,\left\{  a,b,c\right\}  \right\}  $ and $\Phi\left(
G\right)  =\left\{  \left\{  z\right\}  ,\left\{  x,y\right\}  ,\left\{
a,b,c\right\}  \right\}  $. }%
\label{fig121212}%
\end{figure}

For instance, the graphs from Figure \ref{fig121212}
have:$\ I(G;x)=1+6x+9x^{2}+2x^{3}$, while%

\begin{align*}
I(G\circ2K_{1};x)  &  =\left(  1+x\right)  ^{6}\left(  1+12x+48x^{2}%
+76x^{3}+48x^{4}+12x^{5}+x^{6}\right)  =\\
&  =\left(  1+x\right)  ^{5}\cdot I(\Gamma(G);x)=\left(  1+x\right)  ^{6}\cdot
I(\Phi(G);x),\\
I(\Gamma(G);x)  &  =1+13x+60x^{2}+124x^{3}+124x^{4}+60x^{5}+13x^{6}+x^{7},\\
I(\Phi(G);x)  &  =1+12x+48x^{2}+76x^{3}+48x^{4}+12x^{5}+x^{6}.
\end{align*}

Clearly, for every $k\leq\mu\left(  G\right)  $ there exists a clique cover
containing $k$ non-trivial cliques, namely, edges. Consequently, we obtain the following.

\begin{theorem}
For every graph $G$ and for each non-negative integer $k\leq\mu\left(
G\right)  $, one can build a graph $H$, such that: $G$ is a subgraph of $H$,
$I\left(  H;x\right)  $ is symmetric, and $I\left(  G\circ2K_{1};x\right)
=\left(  1+x\right)  ^{k}\cdot I\left(  H;x\right)  $.
\end{theorem}

The characterization of graphs whose independence polynomials are symmetric is
still an open problem \cite{St98}.

Let us mention that there are non-isomorphic graphs with the same independence
polynomial, symmetric or not. For instance, the graphs $G_{1}$, $G_{2}$,
$G_{3}$, $G_{4}$ presented in Figure \ref{fig633} are non-isomorphic, while%
\begin{align*}
I(G_{1};x)  &  =I(G_{2};x)=1+5x+5x^{2}\text{, and }\\
I(G_{3};x)  &  =I(G_{4};x)=1+6x+10x^{2}+6x^{3}+x^{4}.
\end{align*}

\begin{figure}[h]
\setlength{\unitlength}{1cm}\begin{picture}(5,1.2)\thicklines
\multiput(7.2,0)(1,0){3}{\circle*{0.29}}
\multiput(7.2,1)(1,0){3}{\circle*{0.29}}
\put(7.2,0){\line(1,0){2}}
\put(7.2,0){\line(0,1){1}}
\put(7.2,0){\line(1,1){1}}
\put(8.2,0){\line(1,1){1}}
\put(6.7,0.5){\makebox(0,0){$G_{3}$}}
\multiput(10.6,0)(1,0){2}{\circle*{0.29}}
\multiput(10.6,1)(1,0){2}{\circle*{0.29}}
\multiput(12.6,0)(0,1){2}{\circle*{0.29}}
\put(10.6,0){\line(1,0){1}}
\put(10.6,0){\line(0,1){1}}
\put(10.6,0){\line(1,1){1}}
\put(10.6,1){\line(1,0){1}}
\put(11.6,0){\line(0,1){1}}
\put(10.1,0.5){\makebox(0,0){$G_{4}$}}
\multiput(3.9,0)(1,0){3}{\circle*{0.29}}
\multiput(4.9,1)(1,0){2}{\circle*{0.29}}
\put(3.9,0){\line(1,0){2}}
\put(3.9,0){\line(1,1){1}}
\put(4.9,0){\line(0,1){1}}
\put(5.9,0){\line(0,1){1}}
\put(3.6,0.5){\makebox(0,0){$G_{2}$}}
\multiput(0.8,0)(1,0){3}{\circle*{0.29}}
\multiput(1.8,1)(1,0){2}{\circle*{0.29}}
\put(0.8,0){\line(1,0){2}}
\put(0.8,0){\line(1,1){1}}
\put(1.8,1){\line(1,0){1}}
\put(2.8,0){\line(0,1){1}}
\put(0.5,0.5){\makebox(0,0){$G_{1}$}}
\end{picture}
\caption{Non-isomorphic graphs.}%
\label{fig633}%
\end{figure}
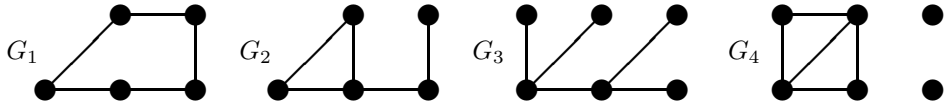

Recall that a graph having at most two vertices with the same degree is called
\textit{antiregular} \cite{Merris2001}. It is known that for every positive
integer $n\geq2$ there is a unique connected antiregular graph of order $n$,
denoted by $A_{n}$, and a unique non-connected antiregular graph of order $n$,
namely $\overline{A_{n}}$ \cite{Behzad}. In \cite{LevMan2010} we showed that
the independence polynomial of the antiregular graph $A_{n}$ is:%
\[
I(A_{2k-1};x)=\left(  1+x\right)  ^{k}+\left(  1+x\right)  ^{k-1}-1\text{ and
}I(A_{2k};x)=2\cdot\left(  1+x\right)  ^{k}-1,\quad k\geq1.
\]

Let us mention that $I(A_{2k};x)=I(K_{k,k};x)$ and $I(A_{2k-1};x)=I(K_{k,k-1}%
;x)$, where $K_{m,n}$ denotes the complete bipartite graph on $m+n$ vertices.
Notice that the coefficients of the polynomial
\[
I(A_{2k};x)=2\cdot\left(  1+x\right)  ^{k}-1=\sum\limits_{j=0}^{k}s_{j}x^{j}%
\]
satisfy $s_{j}=s_{k-j}$ for $1\leq j\leq\left\lfloor k/2\right\rfloor $, while
$s_{0}\neq s_{k}$, i.e., $I(A_{2k};x)$ is \textquotedblleft\textit{almost
symmetric}\textquotedblright.

\begin{problem}
Characterize graphs whose independence polynomials are almost symmetric.
\end{problem}

It is known that the product of a polynomial $P\left(  x\right)
=\sum\limits_{k=0}^{n}a_{k}x^{k}$ and its reciprocal $Q\left(  x\right)
=\sum\limits_{k=0}^{n}a_{n-k}x^{k}$ is a symmetric polynomial. Consequently,
if $I(G_{1};x)$ and $I(G_{2};x)$ are reciprocal polynomials, then the
independence polynomial of $G_{1}\cup G_{2}$ is symmetric, because $I\left(
G_{1}\cup G_{2};x\right)  =I(G_{1};x)\cdot I(G_{2};x)$.

\begin{problem}
Describe families of graphs whose independence polynomials are reciprocal.
\end{problem}

\end{document}